\documentclass{fundam}

\usepackage{hyperref}
\publyear{24}
\papernumber{2174}
\volume{191}
\issue{2}

\finalVersionForARXIV
\usepackage{url} % takes care of hyperlinks, preferred over hyperref
\usepackage[ruled,lined]{algorithm2e}% provides Algorithm environment
\usepackage{graphicx}% allows for inclusion of graphic files (figures)
\usepackage{tikz}
\usetikzlibrary{automata, positioning, arrows}
\usepackage{cite}
\usepackage{amssymb}

\begin{document}

	\title{Closeness and Vertex Residual Closeness of Harary Graphs}

	\author{Hande Tuncel Golpek\thanks{Address for correspondence:  Dokuz Eylul University, Maritime Faculty,
                     Buca Izmir Turkey.}\thanks{The authors
                    are grateful to the anonymous referees for their comments and suggestions.
                     \newline \newline
                    \vspace*{-6mm}{\scriptsize{Received August 2023; \ accepted April 2024.}}}
                    \\
		Maritime Faculty \\
		Dokuz Eylul University\\
        Tinaztepe Campus, Buca, Izmir, Turkey\\
		hande.tuncel@deu.edu.tr
		\and Aysun Aytac$^\dag$
         \\
		Department of Mathematics \\
		Ege University \\
		Bornova, Izmir, Turkey }
	
	\maketitle
	
	\runninghead{H.T. Golpek and A. Aytac}{Closeness and Vertex Residual Closeness of Harary Graphs}
	
	\begin{abstract}
		Analysis of a network in terms of vulnerability is one of the most significant problems. Graph theory serves as a valuable tool for solving complex network problems, and there exist numerous graph-theoretic parameters to analyze the system's stability. Among these parameters, the closeness parameter stands out as one of the most commonly used vulnerability metrics. Its definition has evolved to enhance the ease of formulation and applicability to disconnected structures. Furthermore, based on the closeness parameter, vertex residual closeness, which is a newer and more sensitive parameter compared to other existing parameters, has been introduced as a new graph vulnerability index by Dangalchev. In this study, the outcomes of the closeness and vertex residual closeness parameters in Harary Graphs have been examined. Harary Graphs are well-known constructs that are distinguished by having $n$ vertices that are $k$-connected with the least possible number of edges.
	\end{abstract}
	

%\maketitle
\section{Introduction}

Network analysis is one of the most important problems in the development of
computer science. The key point in network analysis is to determine
centrality. In other words, it is to determine which node has a critical place in the network. This is directly related to the reliability of the
structure. Graph theory has many solution techniques and approaches in this
regard. There are many studies on this analysis in the literature, and these
studies provide links between graph theory and computer science and its
applications. Therefore, graph theory plays an important role in the
analysis of the robustness of networks.

There are many graph-theoretic parameters related to the analysis of complex
networks. Among the parameters associated with centrality, closeness stands out and has undergone various interpretations over the years. The initial closeness approach, when first introduced, could not provide an interpretation for disconnected structures \cite{Freeman}. However, later, through the concept introduced by Latora and Marchiori \cite{LatoraMarchiori}, the possibility of its application to disconnected structures emerged. Building on these definitions, in \cite{Dangalchev}, Dangalchev proposed a formulae for closeness. This approach provides ease in terms of formulation. In this paper, we will utilize them. Dangalchev's useful closeness formula for vertex $i$ is $C(i)=\sum\limits_{j\neq i}\frac{1}{2^{d(i,j)}}$ where $d(i,j)$ represents the distance between vertices $i$ and $j$. Additionally, the concept of residual closeness, a more sensitive parameter based on this definition of closeness, also emerged simultaneously from Dangalchev again \cite{Dangalchev}. The key point here is to determine how the removal of a vertex from the graph impacts the graph's vulnerability. To calculate the closeness value after the removal of vertex $k$, denoted as $C_{k}$, we use the equation  $C_{k}=\sum\limits_{i}\sum\limits_{j\neq i}\frac{1}{2^{d_{k}(i,j)}}$ where $%
d_{k}(i,j)$ is the distance between vertices $i$ and $j$ after the removal of vertex $k$. Moreover, the vertex residual closeness, denoted as $R,$ is defined as $R=%
\underset{k}{\min }\{C_{k}\}$. We refer readers to \cite{AytacOdabas1,AytacOdabas2,AytacOdabas3,AytacOdabas4,BerberlerYigit,AytacVTuraci,Dangalchev,Dangalchev2,Dangalchev3,Dangalchev4,OdabasAytac,
TufanAytac,TufanAytacV,TuraciOkten} for more detailed information and advanced knowledge about closeness and residual closeness parameters. Both closeness and residual closeness
parameters are based on the concept of distance in graphs, and these
measurements are usable in disconnected structures. The term residual closeness mentioned throughout this study should be understood as vertex residual closeness.

For notation and terminology, we utilized two books \cite{Chartrand} and \cite%
{Harary}. In this paper, the graph $G$ is considered a simple, finite, and
undirected graph with vertex set $V(G)$ and edge set $E(G)$. The \textit{%
open neighborhood} of any vertex in $V(G)$, denoted by $N(v)$ $=\{u\in V|$ $%
(uv)\in E(G)\}$. \textit{The degree of a vertex} $v$ denoted by $deg(v)$, is the number of its neighborhood. The distance between two vertices $u$ and $v$ is several edges in the shortest path between them (also it is called $u$-$v$ geodesic),
abbreviated by $d(u,v)$. In addition, three of the most important parameters
of distance in the graph are eccentricity, diameter, and radius. In a
connected graph, \textit{the eccentricity,} $\epsilon (v)$, of a vertex $v$
is the greatest distance between $v$ and any other vertex. \textit{The
diameter }of a graph $G$ is the maximum eccentricity of any vertex in the
graph. It can be symbolized by $diam(G)=\max\limits_{v\in V(G)}\epsilon
(v)=\max\limits_{v\in V(G)}\max\limits_{u\in V(G)}d(u,v)$ and the \textit{%
radius of a graph} $G$, $rad(G)$, is the minimum eccentricity of any vertex.

\medskip
As the number of edges increases, so does the connectivity. This can have negative economic implications in practical applications. This circumstance has had an impact on the emergence of the Harary Graph. Harary Graph concept was introduced in 1962 by F. Harary \cite{Harary62}.
It is a k-connected graph on $n$ vertices that has a degree of at least $k$
and $\lceil\frac{kn}{2}\rceil$ edges. Harary Graph structure constructs in three cases depending on $k$ and $n$ parities. For $2<k<n$, the Harary Graph denoted by $H_{k,n}$ on $n$ vertices
is defined as follows using the expression of West
\cite{West}. For $n$ equally spaced vertices are placed around a
circle. If $k$ is even, $H_{k,n}$ is formed by joining each vertex to the
nearest $k/2$ vertices in each direction around the circle. If $k$ is odd
and $n$ is even, $H_{k,n}$ is obtained by joining each vertex to the nearest
$(k-1)/2$ vertices in a clockwise and counterclockwise direction around the
circle and to the \textit{diametrically opposite} vertex $(n/2)$. For these
two cases, $H_{k,n}$ is $k-$regular and there is an automorphism between
any two vertices $u,v\in V.$ In other words, the graph is vertex-transitive. If both $k$ and $n$ are odd, $H_{k,n}$ is structured as follows. The vertex set
of $H_{k,n}$ is labeled by $0,1,...,n-1$ and $H_{k,n}$ is formed from $%
H_{k-1,n}$ by joining vertex $i$ to vertex $i$ $+(n-1)/2$ for $0\leq i\leq
(n-1)/2$. According to this construction, vertex $(n-1)/2$ is adjacent to
both vertex 0 and vertex $n-1.$ Therefore, the graph is not vertex-transitive
for the case where both $k$ and $n$ are odd. In an Harary Graph, when $k$ is even, if we visit node $(i+k/2)$ or $(i-k/2)$ for $\pmod n$ from vertex $i$ , then it is called a \textit{city-tour}; when $k$ is odd, if we visit node $(i+(k-1)/2)$ or $(i-(k-1)/2)$ for $\pmod n$ from vertex $i$ , then it is called a \textit{city-tour}, otherwise it is called a \textit{village-tour} \cite{Tanna}. Also, the diameter of Harary Graphs $%
H_{k,n}$ is given as follows \cite{Tanna}.
\begin{itemize}
\item If $k$ is even then
\end{itemize}
$$ diam(H_{k,n})=
\begin{cases}
\left\lfloor \frac{n}{k}\right\rfloor, &\mbox{if } \text{$n\equiv 1\;(\bmod \ k) $ }\\
\\[-4pt]
\left\lceil \frac{n}{k}\right\rceil, & \mbox{ otherwise}%
\end{cases}%
$$
\begin{itemize}
\item Let $n$ be even and $k>3$ be odd then
\end{itemize}
$$diam(H_{k,n})=
\begin{cases}
\left\lceil \frac{n}{2k-2}\right\rceil $+1$ ,&  \text{\mbox{if} $diam(H_{k-1,n})$ is even}
\\&\quad%
\text{and $n\not\equiv 2\;(\bmod \ (k-1))$ }\\
\\[-4pt]
\left\lceil \frac{n}{2k-2}\right\rceil ,&\quad \mbox{ otherwise}%
\end{cases}%
$$
\begin{itemize}
\item Let $n$ be even and $k=3$ then
\end{itemize}
\begin{equation*}
diam(H_{3,n})=\left\lceil \frac{n}{4}\right\rceil
\end{equation*}
\begin{itemize}
	\item Let $n$ be odd and $k=3$ then
\end{itemize}
\begin{equation*}
	diam(H_{3,n})=\left\lceil \frac{n+1}{4}\right\rceil
\end{equation*}

\begin{itemize}
\item If both $n$ and $k>3$ are odd then
\end{itemize}
$$diam(H_{k,n})=
\begin{cases}
  \left\lceil \frac{n}{2k-2}\right\rceil +1 , & \mbox{if $(n-k-1)\equiv1\;(\bmod \ {2(k-1)})$}\\
  \\[-4pt]
  \left\lceil \frac{n}{2k-2}\right\rceil, & \mbox{otherwise}.
\end{cases}
$$

Now, we can define the Consecutive Circulant Graph as follows.
Let $[l]=\{1,2,...,l\}$. The Consecutive Circulant Graph, $C_{n},_{[\ell ]}$
is the set of $n$ vertices such that vertex $v$ is adjacent to vertex $v\pm
i $ \;$\bmod \  n$ for each $i\in \lbrack \ell ]$. Notice that $C_{n,[1]}$
is equivalent to $C_{n}$ and $C_{n},_{\lfloor n/2\rfloor }$ is equivalent to
the complete graph $K_{n}$. \cite{Consecutive}. In the next section, we
construct a link between the Harary Graph and the Consecutive Circulant Graph then we will provide the closeness value.

In this study, we begin by examining the closeness values of both Harary Graphs and their specific form known as Consecutive Circulant Graphs. Due to the form of Harary Graphs, this calculation discussed under three main thoughts according to whether $k$
and $n$ are odd or even numbers. Furthermore, under these main ideas, the
subcases of diameters greater than or less than $2$ are also considered. After calculating the closeness value of the Consecutive circulant
Graphs, the second part of the paper is completed. In the third part,
the vertex residual closeness of Harary Graphs is taken up. As was done when computing the closeness values, the cases were handled individually, taking into account the parity relationships between $k$ and $n$.

\medskip
Additionally, let us provide the geometric series sum formula and its derivative as preliminary information, which serves as a useful tool frequently used in proofs:\\
\begin{align}
	\sum\limits_{k=1}^{h}aX^{k-1}  & =a(X^{0}+...+X^{h-1})=\frac{a(1-X^{h})}%
	{1-X} \label{geo}\\
	\sum\limits_{k=2}^{h}a(k-1)X^{k-2}  & =a\cdot[\frac{hX^{h-1}}{X-1}-\frac{X^{h}%
		-1}{(X-1)^{2}}] \label{derigeo}.\vspace*{1mm}
\end{align}

\section{Closeness of Harary Graphs}

In this section, we evaluated the closeness value of a Harary Graph, denoted as $C(H_{k,n})$. Additionally, by examining the relationship between the Harary Graph and Consecutive Circulant Graph structures, we will reveal the connection between the closeness values.

\begin{theorem}
\label{kiseven}Let $H_{k,n}$ be a Harary Graph on $n$ vertices where $k$ is
even, then
$$ \small{ C(H_{k,n})=
\begin{cases}
  n(k+(\frac{1}{2})^{diam(H_{k,n})}(t-2k)) & ,\mbox{if } (n-1)\equiv t\;(\bmod \  k) \\
  n(k-(\frac{1}{2})^{diam(H_{k,n})}k) & ,\mbox{if} (n-1)\equiv 0 \;(\bmod \  k).
\end{cases}}
$$

\begin{proof}
Let the vertices of the graph be labeled as $\{0,1,...,(n-1)\}$. It should be
considered in two cases according to the value of $(n-1)$ in $(\bmod\ k)$. If $k$ is even then the graph is vertex-transitive, it refers to all vertices having the same closeness value. Therefore, without loss of generality, we can take vertex $0$ as the originator and evaluate the closeness value of vertex $0$.

\medskip
    If $(n-1)\equiv t (\bmod\ k),$ $t\neq 0$ then the distance
between $0$ and vertex $\frac{k}{2}\cdot(diam(H_{k,n})-1)$ and vertex $n-\frac{k%
}{2}\cdot(diam(H_{k,n})-1)$ at a clockwise and counterclockwise, respectively,
is $(diam(H_{k,n})-1)$. Thus, the remaining $t$ vertices from $0$ is $(
diam(H_{k,n})-1)$ city-tour plus $1$ unit or $diam(H_{k,n})$ city-tour
distance. Hence, there are $k$ vertices in the distance $1,2,...,(diam(H_{k,n})-1)$ from vertex $0$ and there are $t$ vertices at distance $diam(H_{k,n})=\left%
\lceil \frac{n}{k}\right\rceil .$ Thus, we get%

\begin{equation*}
C(0)=\sum\limits_{i=1}^{diam(H_{k,n})-1}k(\frac{1}{2})^{i}+(t)(\frac{1}{2}%
)^{diam(H_{k,n})}.
\end{equation*}%
Using geometric series formula (\ref{geo}), we have%
\begin{equation*}
C(0)=k+(\frac{1}{2})^{diam(H_{k,n})}(t-2k).
\end{equation*}%
Therefore, the closeness value of a graph with $n$ vertices where $n-1\equiv
t\;(\bmod \  k)$ is%
\begin{equation*}
C(H_{k,n})=n(k+(\frac{1}{2})^{diam(H_{k,n})}(t-2k)).
\end{equation*}

    If $(n-1)\equiv 0 \pmod k$ then there are $k$ vertices in
distance $1,2,...,diam(H_{k,n})$ from vertex $0$ where $diam(H_{k,n})=\left%
\lfloor \frac{n}{k}\right\rfloor $ for this case. The closeness formula for this
case for vertex $0$ can be obtained as%
\begin{equation*}
C(0)=\sum\limits_{i=1}^{diam(H_{k,n})}k(\frac{1}{2})^{i}.
\end{equation*}%
Using geometric series summation, we have%
\begin{equation*}
C(0)=k-(\frac{1}{2})^{diam(H_{k,n})}(k).
\end{equation*}%
Therefore, the closeness value of a graph with $n$ vertices where $(n-1)\equiv
0 \;(\bmod \  k) $ is%
\begin{equation*}
C(H_{k,n})=n(k-(\frac{1}{2})^{diam(H_{k,n})}k).
\end{equation*}%
Thus, the theorem holds. \end{proof}
\end{theorem}

\begin{theorem}
$\!$Let $H_{k,n}\!$  be a Harary Graph on  $n$ vertices where $k$ is odd, $n$ is
even and  $diam(H_{k,n})\!\leq\!2$. Then,
\begin{equation*}
C(H_{k,n})=\left\{
\begin{array}{cc}
\frac{n(n-1)}{2} & ,\text{if }diam(H_{k,n})=1 \\
n(\frac{k}{2}+\frac{(n-k-1)}{2^{2}}) & ,\text{if }diam(H_{k,n})=2%
\end{array}%
\right.
\end{equation*}
\end{theorem}

\begin{proof}
The graph is vertex-transitive for the case $k$ is odd and $n$ is even. For $diam(H_{k,n})=1$, the distance of $n$ vertices to $(n-1)$ vertices
is one. Thus, we have%
\begin{equation*}
C(H_{k,n})=\frac{n(n-1)}{2}.
\end{equation*}%
For $diam(H_{k,n})=2,$ let choose vertex $0$ as an originator.
Therefore, there are $k$ vertices at the distance one, and the remaining vertices $(n-k-1)$ are at the distance two to any vertex from vertex $0$ in the graph.
Therefore, we get
\begin{equation*}
C(H_{k,n})=n(\frac{k}{2}+\frac{(n-k-1)}{2^{2}}).
\end{equation*}%
Thus, the proof is completed.
\end{proof}
\begin{remark}When the value of $k$ is even, the structure of the Harary Graph is in a simpler form. Even in the calculation of the diameter, different cases arise when $k$ is odd. Therefore, to facilitate the expression of the number of city-tours and distances between vertices for the case where $k$ is odd, we will utilize the $H_{k-1, n}$ graph forms. Consequently, employing the value $diam(H_{k-1, n})$ in expressing the theorems will be useful in reducing complexity in proofs.
	\end{remark}

\begin{theorem}\label{thm2.3}
Let $H_{k,n}$ be a Harary Graph on $n$ vertices where $k>3$\ is odd $n$\ is
even and \linebreak $diam(H_{k,n})>2$ and $n\equiv t \;(\bmod \ {(k-1)}),$ $%
t\neq 0,2$. Then,

\begin{itemize}
  \item if $diam(H_{k-1,n})$ is odd, then
\end{itemize}
\begin{equation*}
%\resizebox{0.5\textwidth}{!}{
C(H_{k,n})=\frac{nk}{2}+n(k-1)(1-\frac{1}{2^{diam(H_{k,n})-2}})+\frac{%
n((k+t-3))}{2^{diam(H_{k,n})}}
%}
\end{equation*}

\begin{itemize}
  \item if $diam(H_{k-1,n})$ is even, then
\end{itemize}
\begin{equation*}
%\resizebox{0.5\textwidth}{!}{
C(H_{k,n})=\frac{nk}{2}+n(k-1)(1-\frac{1}{2^{diam(H_{k,n})-2}})+\frac{n(t-2)%
}{2^{diam(H_{k,n})}}.
%}
\end{equation*}
\end{theorem}

\begin{proof}
Let $n\equiv t \;(\bmod \ {(k-1)})$ and $diam(H_{k,n})>2.$ Assume
that vertices of $H_{k,n}$ are labeled as $0,...,(n-1)$. In order to construct $%
H_{k,n}$ when $k$ is odd, join vertex $0$ and $n/2$ in the form of $%
H_{k-1,n}.$ In this situation, there are $\frac{k-3}{2}+1$ options for $t$
such as$\ t=0,2,4,...,(k-3).$\\
 Let us consider the values of $t\neq 0,2.$ Since
the graph is vertex-transitive for this case, without loss of generality, we can choose vertex $0$ as the origin vertex in order to find\ the closeness value of a vertex. In the form of $H_{k-1,n},$ the distance between $0$ and vertex
$\frac{k-1}{2}\cdot(diam(H_{k,n})-1)$ in clockwise direction and using the
diametrically opposite direction between vertex $0$, and vertex $\frac{n}{2}-\frac{k-1}{2}\cdot(diam(H_{k,n})-2)$ is $(diam(H_{k,n})-1)$. Analogously, we get the same results for the counterclockwise direction from vertex $0.$ Therefore, the distance between for the total remaining $(t-2)$ vertices in both directions and vertex $0$ is $(diam(H_{k,n})-1)$ city-tour plus $1$ unit or $diam(H_{k,n})$
city-tour that is $diam(H_{k-1,n})=\left\lceil \frac{n}{k-1}\right\rceil $
away from vertex $0.$\ When vertices $0$ and $n/2$ join to get $H_{k,n}$, $%
diam(H_{k,n})=\left\lceil \frac{n}{2k-2}\right\rceil +1$, see
\cite{Tanna}.

\medskip
    If $diam(H_{k-1,n})$ is odd, there are $(k-1)+1=k$ vertices one
distance from vertex $0$, there are $2(k-1)$ vertices at most $%
(diam(H_{k,n})-1)$ distance and at least 2 distance from vertex 0 and there
are $(k-1)+(t-2)$ vertices $diam(H_{k,n})$ distance from vertex $0$. Thus,
the total value for $n$ vertices is%
\begin{eqnarray*}
C(H_{k,n}) &=&n[k\frac{1}{2}+2(k-1)\frac{1}{2^{2}}+...+2(k-1)\frac{1}{2^{diam(H_{k,n})-1}}+\frac{(k-1)+(t-2)}{2^{diam(H_{k,n})}}] \\
&=&\frac{nk}{2}+n(k-1)(1-\frac{1}{2^{diam(H_{k,n})-2}})+\frac{n((k+t-3)}{2^{diam(H_{k,n})}}.
\end{eqnarray*}

    If $diam(H_{k-1,n})$ is even, there are $(k-1)+1=k$ vertices one
distance from vertex $0$, there are $2(k-1)$ vertices in $2,3$,...,($%
diam(H_{k,n})-1)$ distance from $0$ and there are $(t-2)$ vertices $%
diam(H_{k,n})$ distance from vertex $0$. Thus, the total value for $n$ vertices
is%

\begin{eqnarray*}
  C(H_{k,n}) &=& n[k\frac{1}{2}+2(k-1)\frac{1}{2^{2}}+... + 2(k-1)\frac{1}{2^{diam(H_{k,n})-1}}+\frac{(t-2)}{2^{diam(H_{k,n})}}] \\
 &=& \frac{nk}{2}+n(k-1)(1-\frac{1}{2^{diam(H_{k,n})-2}})+\frac{n(t-2)}{2^{diam(H_{k,n})}}.
\end{eqnarray*}
Thus, the result is obtained.\end{proof}

\begin{theorem}
Let $H_{k,n}$ be a Harary Graph on $n$ vertices where $k>3$ is odd, $n$ is
even, and $diam(H_{k,n})>2$ and $n\equiv 2 \;(\bmod \ {(k-1)})$. Then,
\begin{itemize}
\item if $diam(H_{k-1,n})$ is odd, then
\end{itemize}
\begin{equation*}
C(H_{k,n})=\frac{nk}{2}+n(k-1)(1-\frac{3}{2^{diam(H_{k,n})}})
\end{equation*}

\begin{itemize}
\item if $diam(H_{k-1,n})$ is even, then
\end{itemize}
\begin{equation*}
C(H_{k,n})=\frac{nk}{2}+n(k-1)(1-\frac{1}{2^{diam(H_{k,n})-1}}).
\end{equation*}
\end{theorem}

\begin{proof}
Let $n\equiv2\;(\bmod \ {(k-1)})$. Therefore, $diam(H_{k-1,n})=\left\lceil \frac{n}{k-1}\right\rceil\in\mathbb{Z}$. Similar to the proof of Theorem
\ref{thm2.3}%
, without loss of generality, let us consider the vertex $0$ as originated.
Constructing $H_{k-1,n},$ there are exactly $(diam(H_{k-1},_{n})-1)$ city- tour plus one unit to reach vertex $n/2.$

\medskip
    If $diam(H_{k-1,n})$ is odd then there are exactly even numbers of the city-tours between vertices $0$ and $n/2$ in graph $H_{k-1,n}.$ Since $(n-2)\equiv 0 \;(\bmod \ {(k-1)})$, in the form of $H_{k,n},$ $0$ and $n/2$ are
diametrically opposite vertices. Thus, for the vertex $0$, there are $k$ adjacent vertices, and there are $2(k-1)$ vertices in $2, 3, ..., diam(H_{k,n}-1)$ distance, and there are $k$ vertices at $diam(H_{k,n})$ distance. Hence, the total value for $n$ vertices is%
\begin{eqnarray*}
C(H_{k,n}) &=&n[\frac{k}{2}+\frac{2(k-1)}{2^{2}}+...+\frac{2(k-1)}{2^{diam(H_{k,n})-1}}+\frac{(k-1)}{2^{diam(H_{k,n})}}] \\
&=&\frac{nk}{2}+\frac{n(k-1)}{2}[1+\frac{1}{2}+...+\frac{1}{2^{diam(H_{k,n})-3}}]+\frac{n(k-1)}{2^{diam(H_{k,n})}}\\
&=&\frac{nk}{2}+n(k-1)[1-\frac{3}{2^{diam(H_{k,n})}}].
\end{eqnarray*}

    If $diam(H_{k-1,n})$ is even, then there are exactly an odd number of
city-tours between vertex $0$ and $n/2$ in graph $H_{k-1,n}.$ Since $%
n-2\equiv 0 \;(\bmod \ {(k-1)})$, in the form of $H_{k,n},$ $0$ and $n/2$ are
diametrically opposite vertices and there are exactly $2(k-1)$ vertex in
every distance except one from vertex $0$. Thus, the total value for $n$
vertices is\vspace*{-3mm}%
\begin{eqnarray*}
C(H_{k,n}) &=&n[k\frac{1}{2}+2(k-1)\frac{1}{2^{2}}+...+2(k-1)\frac{1}{2^{diam(H_{k,n})-1}}+2(k-1)\frac{1}{2^{diam(H_{k,n})}}] \\
&=&\frac{nk}{2}+n(k-1)(1-\frac{1}{2^{diam(H_{k,n})-1}}).
\end{eqnarray*}%
Therefore, the proof is completed.\end{proof}

\begin{theorem}
Let $H_{k,n}$ be a Harary Graph on $n$ vertices where $k>3$ is odd, $n$ is
even, and $diam(H_{k,n})>2$ and $n\equiv 0 \;(\bmod \ {(k-1)})$. Then,

\begin{itemize}
\item if $diam(H_{k-1,n})$ is odd, then
\end{itemize}
\begin{equation*}
%\resizebox{0.5\textwidth}{!}{
C(H_{k,n})=\frac{nk}{2}+n(k-1)(1-\frac{1}{2^{diam(H_{k,n})-2}})+\frac{2n(k-2)}{2^{diam(H_{k,n})}}
%}
\end{equation*}
\begin{itemize}
\item if $diam(H_{k-1,n})$ is even, then
\end{itemize}
\begin{equation*}
%\resizebox{0.5\textwidth}{!}{
C(H_{k,n})=\frac{nk}{2}+n(k-1)(1-\frac{1}{2^{diam(H_{k,n})-2}})+\frac{n(k-3)}{2^{diam(H_{k,n})}}.
%}
\end{equation*}
\end{theorem}

\begin{proof}
Let $n\equiv 0\;(\bmod \ {(k-1)})$. Therefore, $diam(H_{k-1},_{n})=$ $\frac{n}{k-1}\in\mathbb{Z}
%EndExpansion
.$ Similar to the proof of Theorem
\ref{thm2.3}, without loss of generality, let's consider the distance of vertex $0$ from all other vertices. According to form of $H_{k-1,n},$ there are $\frac{n}{k-1}$ city-tour from vertex $0$ to vertex $n/2$. Vertices $0$ are $n/2$ diametrically opposite in $H_{k,n}$, we get $diam(H_{k,n})=\left\lceil \frac{n%
}{2k-2}\right\rceil +1$. Thus, there are $2\cdot(\frac{(k-1)}{2}-1)=k-3$
vertices occur diametrically distance from vertex $0$ utilizing $%
(diam(H_{k-1},_{n})-1)$ city-tour and one unit.

\medskip
    If $diam(H_{k-1,n})$ is odd there are $(k-1)+1=k$ vertices one
distance from vertex $0$, there are $2(k-1)$ vertices at most $%
diam(H_{k,n})-1$ distance and at least $2$ distance from vertex $0$ and
there are $(k-1)+(k-3)$ vertices $diam(H_{k,n})-$distance from vertex $0$.
Thus, the total value for $n$ vertices is%
\begin{align*}
\begin{split}
   C(H_{k,n}) =n[k\frac{1}{2}+2(k-1)\frac{1}{2^{2}}+...+2(k-1)\frac{1}{2^{diam(H_{k,n})-1}}+\frac{(k-1)+(k-3)}{2^{diam(H_{k,n})}}].
\end{split}
\end{align*}

    If $diam(H_{k-1,n})$ is even, there are $(k-1)+1=k$ vertices one
distance from vertex $0$, there are $2(k-1)$ vertices at $2,3,...,%
(diam(H_{k,n})-1)$ distance from $0$, and there are totally $(k-3)$ vertices $diam(H_{k,n})$ distance from vertex $0$ in both directions. Thus, the total value for $n$ vertices is\vspace*{-1.6mm}
\begin{eqnarray*}
\begin{split}
   C(H_{k,n}) =n[k\frac{1}{2}+2(k-1)\frac{1}{2^{2}}+...+2(k-1)\frac{1}{2^{diam(H_{k,n})-1}}+\frac{(k-3)}{2^{diam(H_{k,n})}}].
\end{split}
\end{eqnarray*}%
Then the proof is completed.\end{proof}

\begin{corollary}
Let $H_{k,n}$ be a Harary Graph on $n$ vertices where $k=3$ and $n$ is
even, and \linebreak $diam(H_{k,n}) >2$.

\begin{itemize}\vspace*{-1mm}
\item if $diam(H_{2,n})$ is odd, then
\end{itemize}
\begin{equation*}
C(H_{k,n})=\frac{7n}{2}-\frac{6n}{2^{diam(H_{k,n})}}
\end{equation*}
\begin{itemize}\vspace*{-1mm}
\item if $diam(H_{2,n})$ is even, then
\end{itemize}
\begin{equation*}
C(H_{k,n})=\frac{7n}{2}-\frac{2n}{2^{diam(H_{k,n})-1}}.
\end{equation*}
\end{corollary}

\begin{proof}
  If $k=3$ then $n\equiv 0 \;(\bmod \ {(k-1)})$. However, for $diam(H_{2,n})$ is even, this corollary is not a special case of the
  previous theorem due to the calculation of diameter.

\medskip
  If $diam(H_{2,n})$ is odd there are three vertices one
distance from vertex $0$, there are four vertices at most $%
(diam(H_{k,n})-1)$ distance and there are two vertices at $diam(H_{k,n})-$distance from vertex $0$.
Thus, the total value for $n$ vertices is%
\begin{align*}
\begin{split}
   C(H_{k,n}) &=\frac{3n}{2}+4n[\frac{1}{2^{2}}+...+\frac{1}{2^{diam(H_{k,n})-1}}]+\frac{2}{2^{diam(H_{k,n})}}\\
   &=\frac{3n}{2}+2n(1-\frac{3}{2^{diam(H_{k,n})}})\\
   &=\frac{7n}{2}-\frac{6n}{2^{diam(H_{k,n})}}.
\end{split}
\end{align*}

If $diam(H_{k-1,n})$ is even, there are three vertices at one
distance from vertex $0$, there are four vertices at $2,3,...,$$(diam(H_{k,n})$ distance from $0$. Thus, the total value for $n$ vertices is
\begin{eqnarray*}
\begin{split}
   C(H_{k,n}) & =\frac{3n}{2}+4n[\frac{1}{2^{2}}+...+ \frac{1}{2^{diam(H_{k,n})}}]\\
     &=\frac{3n}{2}+2n(1-\frac{1}{2^{diam(H_{k,n})-1}})\\
     &=\frac{7n}{2}-\frac{2n}{2^{diam(H_{k,n})-1}}.
\end{split}
\end{eqnarray*}

\vspace*{-8mm}
\end{proof}\vspace*{1mm}

\begin{theorem}
Let $H_{k,n}$ be a Harary Graph on $n$ vertices where $k$\ is odd $n$\ is
odd and \linebreak $diam(H_{k,n})\leq 2.$ Then,
\begin{equation*}
C(H_{k,n})=\left\{
\begin{array}{cc}
\frac{n(n-1)}{2} & ,\text{if }diam(H_{k,n})=1 \\
\\[-6pt]
\frac{n^2+nk-n+1}{4}& ,\text{if }diam(H_{k,n})=2%
\end{array}%
\right. .
\end{equation*}
\end{theorem}

\begin{proof}
If $diam(H_{k,n})=1$ then the distance of $n$ vertices to $n-1$ vertices is $
1.$ Thus, we have,
\begin{equation*}
C(H_{k,n})=\frac{n(n-1)}{2}.
\end{equation*}%

\eject
Let $diam(H_{k,n})=2.$ For the vertex labeled by $(n-1)/2$, there are $(k+1)$ vertices that are adjacent to $(n-1)/2$, and there are $(n-1)-(k+1)=n-k-2$ vertices at a distance of two. For the remaining $(n-1)$ vertices, there are $k$ adjacent vertices and $(n-k-1)$ vertices at two distances to them. Therefore, we have
\begin{eqnarray*}
C(H_{k,n}) &=&\frac{k+1}{2}+\frac{n-k-2}{2^{2}}+(n+1)[\frac{k}{2}+\frac{n-k-1}{2^{2}}] \\
&=&\frac{n^2+nk-n+1}{4}.
\end{eqnarray*}%
Thus, the theorem holds.
\end{proof}

\begin{theorem}
\label{nkodd}Let $H_{k,n}$ be a Harary Graph on $n$ vertices where $k\geq3$  is
odd, $n$\ is odd, and \linebreak $diam(H_{k,n})>2$. Then,
\begin{equation*}
\begin{split}
   C(H_{k,n}) =\frac{A+kn+1}{2}+\frac{nt-A-1}{2^{diam(H_{k,n})}}+(k-1)(nB-\frac{diam(H_{k,n})-3}{2}-(\frac{1}{2})^{diam(H_{k,n})-1})
\end{split}
\end{equation*}%
where $(n-k-1)\equiv t \;(\bmod\ 2(k-1))$, $A=(diam(H_{k,n})-2)(k-1)$ and $B=1-(\frac{1}{2})^{diam(H_{k,n})-2}.$
\end{theorem}

\begin{proof}
In this case, the graph is not vertex-transitive \cite{Tanna}. Let $%
V(H_{k,n})=\{0,1,...,n-1\}$ be the vertex set of the graph. The closeness
value is obtained as different as the value of the diameter in the graph.
The first type of closeness belongs to vertex $v=\frac{n-1}{2},$ the other
types of closeness come from vertices $v=\frac{n-1}{2}\pm ((\frac{k-1}{2}%
)j+i),$ for$~0\leq j\leq (diam(H_{k,n})-3)$ and $1\leq i\leq (\frac{k-1}{2})$
$.$ This means that there are ($diam(H_{k,n})-2)$ vertex groups include $(k-1)$ vertices with the same closeness value on both sides of the vertex $\frac{n-1%
}{2}$. The other closeness value gets from remaining $n-(k-1)(diam(H_{k,n})-2)-1$
vertices. In addition, we can create a set, denoted by $RM,$ that includes
these remaining $n-(k-1)(diam(H_{k,n})-2)-1$ vertices$.$ Here, there are $t$ vertices at distance diameter to vertices in the set $RM$ where $(n-k-1)\equiv t(\bmod\ 2(k-1))$.

Let us choose vertex $0$ from $RM.$ Since there are $%
(n-k-1)$ vertices in the graph at $2,3,...\linebreak diam(H_{k,n})$ distances from
vertex $0$ and we can make $2,3,...(diam(H_{k,n})-1)$ city-tour from vertex $%
0, $ each city-tour covers $2(k-1)$ vertices. Therefore remaining $t$
number of vertices can be reached by $diam(H_{k,n})-1$ city-tour plus one
village tour.

\medskip
    For the vertex $v=\frac{n-1}{2},$ there are $\ k+1$ vertices in
distance one, there are $2(k-1)$ vertices in each distance $%
2,3,...,(diam(H_{k,n})-1)$\ and since there are two diametrically opposite
vertices of $v=\frac{n-1}{2},$ there are $t-1$ vertices at distance $%
diam(H_{k,n}).$ Hence, the closeness value of $v=\frac{n-1}{2}$ is

\begin{equation}
C(v)=\frac{k+1}{2}+\sum\limits_{i=2}^{diam(H_{k,n})-1}\frac{2(k-1)}{2^{i}}+%
\frac{t-1}{2^{diam(H_{k,n})}}.
\end{equation}
Using the eqution (\ref{geo}), we have
\begin{equation}
\begin{split}
   C(\frac{n-1}{2}) =\frac{k+1}{2}+(k-1)(1-(\frac{1}{2})^{diam(H_{k,n})-2})+\frac{t-1}{2^{diam(H_{k,n})}}.
\end{split}
\label{eq1}
\end{equation}

    For each index $j$ where $0\leq j\leq (diam(H_{k,n})-3),$ the vertices
$v=\frac{n-1}{2}\pm ((\frac{k-1}{2})j+i)$ for $1\leq i\leq (\frac{k-1}{2})$,
each $j$ group is ($\frac{1}{2^{j+1}}-\frac{1}{2^{j+2}}$) less than the
closeness value of the vertex $\frac{n-1}{2}.$ Then, there are ($%
diam(H_{k,n})-2)$ group with $(k-1)$ vertices. Thus, the total closeness value of these kinds of vertices%
\begin{eqnarray*}
C(v)\hspace*{-2mm}
               &=&\hspace*{-2mm}(k-1)[C(\frac{n-1}{2})-(\frac{1}{2}-\frac{1}{2^{2}})]+(k-1)[C(\frac{n-1}{2})-(\frac{1}{2}-\frac{1}{2^{2}})\\
               \hspace*{-2mm}&-&\hspace*{-2mm}(\frac{1}{2^{2}}-\frac{1}{2^{3}})]\!+...+\!(k-1)[C(\frac{n-1}{2})-(\frac{1}{2}-\frac{1}{2^{2}})\!-...-\!
 (\frac{1}{2^{diam(H_{k,n})-2}}-\frac{1}{2^{diam(H_{k,n})-1}})] \\
\hspace*{-2mm}&=&\hspace*{-2mm}(k-1)[(diam(H_{k,n}-2)C(\frac{n-1}{2})\\
\hspace*{-2mm}&-&\hspace*{-2mm}((diam(H_{k,n}-2))\frac{1}{4}-(diam(H_{k,n}-3))\frac{1}{8}-...-\frac{1}{2^{diam(H_{k,n})-1}})].
\end{eqnarray*}
Then, we get
\begin{eqnarray*}
C(v)=(k-1)[(diam(H_{k,n}-2)C(\frac{n-1}{2})-\sum_{i=2}^{diam(H_{k,n})-1}\frac{diam(H_{k,n})-i}{2^{i}}].
\end{eqnarray*}
Using formula (\ref{derigeo}) for $X=1/2$ and $h=diam(H_{k,n})$, then we have
\begin{equation}
\begin{split}
   C(v)  =(k-1)[(diam(H_{k,n}-2)C(\frac{n-1}{2})-(k-1)(\frac{diam(H_{k,n})-3}{2}+\frac{1}{2^{diam(H_{k,n})-1}})
\end{split}
\label{eq2}
\end{equation}

For the remaining vertex in the set $RM$, there are $k$ vertices at
distance one, there are $2(k-1)$ vertices at each distance $%
2,3,...,(diam(H_{k,n})-1)$ and there are $t$ vertices at distance $%
diam(H_{k,n}).$ We can choose vertex $v$ in $RM$ as vertex $0$. Thus, the closeness value of a vertex $v$ in $RM$ is%
\begin{equation}
C(0)=\frac{k}{2}+\sum\limits_{i=2}^{diam(H_{k,n})-1}\frac{2(k-1)}{2^{i}}+%
\frac{t}{2^{diam(H_{k,n})}}.
\end{equation}%
Applying the formula (\ref{geo}), we have
\begin{equation}
C(0)=\frac{k}{2}+(k-1)(1-(\frac{1}{2})^{diam(H_{k,n})-2})+\frac{t}{%
2^{diam(H_{k,n})}}.  \label{eq3}
\end{equation}%
Also, there are $n-(k-1)(diam(H_{k,n})-2)-1$ vertices of this type in the
graph. Therefore, using equations (\ref{eq1}),(\ref{eq2}) and (\ref{eq3})
the closeness value of $H_{k,n}$ can be expressed as for odd values of $k$
and $n$%
\begin{eqnarray}
  \nonumber C(H_{k,n}) &=& C(\frac{n-1}{2})\\ \nonumber
  &+&(n-(k-1)(diam(H_{k,n})-2)-1)C(0) \\\label{burdenli}
  &+&(k-1)(diam(H_{k,n}-2))C(\frac{n-1}{2})\\\nonumber
  &-&(k-1)[(\frac{(diam(H_{k,n})-3)}{2})+\frac{1}{2^{diam(H_{k,n})-1}}].\nonumber
\end{eqnarray}%
In order to get rid of burden in equation (\ref{burdenli}) in terms of
notation, we can take $A=(diam(H_{k,n})-2)(k-1)$ and $B=1-(\frac{1}{2}%
)^{diam(H_{k,n})-2}$ and substitute $A$ and $B$ into the (\ref{eq1}), (\ref{eq2})
and (\ref{eq3}). Thus, we can express closeness value
\begin{equation*}
\begin{split}
C(H_{k,n})=\frac{A+kn+1}{2}+\frac{nt-A-1}{2^{diam(H_{k,n})}}+(k-1)(nB-\frac{diam(H_{k,n})-3}{2}-(\frac{1}{2})^{diam(H_{k,n})-1}).
\end{split}
\end{equation*}%
The closeness value is yield.
\end{proof}
Before considering the closeness value of the Consecutive Circulant Graph,
the connection can be made with the Harary Graph. Consecutive Circulant Graph
definition corresponds to Harary Graph structure when $k$ is even. Thus
the following result can be obtained from Theorem \ref{kiseven}.

\begin{corollary}
Let $C_{n},_{[\ell ]}$ be Consecutive Circulant Graph and $(n-1)\equiv t \ (mod 2\ell)$ then closeness of $C_{n},_{[\ell ]}$%
\begin{equation*}
C(C_{n},_{[\ell ]})=\left\{
\begin{array}{cc}
n\cdot (2l+(\frac{1}{2})^{diam(C_{n},_{[\ell ]})}(t-4l))&,\text{ if }t\neq 0
\\
2nl\cdot (1-\dfrac{1}{2^{diam(C_{n},_{[\ell ]})}})&,\text{ if }t=0%
\end{array}%
\right.
\end{equation*}%
where $diam(C_{n},_{[\ell ]})=\left\lceil \frac{n}{2l}\right\rceil $ when $%
t\neq 0$ and $diam(C_{n},_{[\ell ]})=\left\lfloor \frac{n}{2l}\right\rfloor $
when $t=0.$
\end{corollary}

\begin{proof}
The definition of Consecutive Circulant Graph overlaps the definition of Harary Graph when $k=2l$. Therefore, the Consecutive Circulant Graph is a Harary Graph for $k$ is even case. Thus, substitute $2l$ into the $k$ $\ $and $%
diam(H_{k,n})=diam(C_{n},_{[\ell ]})$ into the Theorem \ref{kiseven}. We get

\begin{equation*}
C(C_{n},_{[\ell ]})=\left\{
\begin{array}{cc}
n\cdot (2l+(\frac{1}{2})^{diam(C_{n},_{[\ell ]})}(t-4l))&,\text{ if }t\neq 0
\\
2nl\cdot (1-\dfrac{1}{2^{diam(C_{n},_{[\ell ]})}})&,\text{ if }t=0%
\end{array}%
\right..
\end{equation*}%

\vspace*{-6mm}
\end{proof}

\section{Vertex residual closeness of Harary Graphs, $H_{k,n}$}

In this section, we will calculate the (vertex) residual closeness value of the Harary Graph for each situation, which is denoted as $R$. This value corresponds to the minimum value of $C_{k}$ which represents the closeness value after removal of vertex $k$. As a result, we will identify the most sensitive node in the graph for each specific case.

\begin{theorem}\label{diam2}
Let $H_{k,n}$ be a Harary Graph on $n$ vertices where $k>2$ is even and $%
diam(H_{k,n})>2.$ Then, residual closeness value of $H_{k,n}$
\begin{itemize}
  \item If $n\equiv 1\; (\bmod\ k)$, then
\end{itemize}
\begin{equation*}
  R=\frac{n-2}{n}C(H_{k,n})-1+(\frac{1}{2})^{diam(H_{k,n})}\cdot(1+diam(H_{k,n}))
\end{equation*}
\eject
\begin{itemize}
  \item If $n\not\equiv 1\; (\bmod\ k)$, then
\end{itemize}
\begin{equation*}
  R=\frac{n-2}{n}C(H_{k,n})-1+(\frac{1}{2})^{diam(H_{k,n})-1}\cdot diam(H_{k,n})
\end{equation*}
where $C(H_{k,n})$ is closeness value for $k$ is even.
\end{theorem}
\begin{proof}
Since $k$ is even the $H_{k,n}$ is vertex transitive. Therefore, removing any
vertex from the graph will have the same effect on the closeness value.
Without loss of generality, let the vertex $0$ be deleted from the graph.
Thus,%
\begin{eqnarray*}
C_{0} &=&C(H_{k,n})-2\sum\limits_{j\neq 0}\frac{1}{2^{d(0,j)}}-D_{0} \\
&=&\frac{n-2}{n}\cdot C(H_{k,n})-D_{0}
\end{eqnarray*}%
where $D_{0}$ denotes changes in closeness value of $C(H_{k,n})$ after
deleting $0$ and also $C(H_{k,n})=n\sum\limits_{j\neq 0}\frac{1}{2^{d(0,j)}}
$ for $k$ is even. The value of $C(H_{k,n})$ is calculated in Theorem \ref{kiseven}. Let evaluate value of $D_{0}$ in terms of value of $n$ in $(\bmod\ k).$

\medskip
    Observation 1. \ for $D_{0}$: \\
     If $n\equiv 1\; (\!\!\!\mod k)$ then there
are $k$ vertices at distance $1,2,3,...,diam(H_{k,n})$ where $%
diam(H_{k,n})=\left\lfloor \frac{n}{k}\right\rfloor .$ When vertex $0$ is
deleted, the distances between for each vertex $\frac{k}{2}.(j-1)$ and $%
n-\frac{k}{2}i$  where $1\leq i\leq (diam(H_{k,n})-j+1)$ and $2\leq j\leq
diam(H_{k,n})$ increase by 1. Hence, $D_{0}=\sum%
\limits_{i=2}^{diam(H_{k,n})}\frac{i-1}{2^{i}}.$ The summation can be expanded
as%

\begin{equation*}
\sum\limits_{i=2}^{diam(H_{k,n})}\frac{i-1}{2^{i}}=(\frac{1}{2^{2}}+\frac{2}{%
2^{3}}+...+\frac{diam(H_{k,n})-1}{2^{diam(H_{k,n})}})
\end{equation*}%
and the summation can be modified as%
\begin{equation}
=\frac{1}{2^{2}}(1+\frac{2}{2}+\frac{3}{2^{2}}+...+\frac{diam(H_{k,n})-1}{%
2^{diam(H_{k,n})-2}})  \label{sum1}
\end{equation}%
and the summation result in (\ref{sum1}) is obtained by using equation (\ref{derigeo}), substituting $X=\frac{1}{2}$ and $h=diam(H_{k,n})$ into
the equation. Thus, we get%
\begin{equation}
=\frac{1}{2^{2}}(\frac{diam(H_{k,n})\cdot\frac{1}{2}^{(diam(H_{k,n})-1)}}{(1/2)-1%
}-\frac{\frac{1}{2}^{diam(H_{k,n})}-1}{((1/2)-1)^{2}})
\end{equation}%
Thus, we get
\begin{equation}
\sum\limits_{i=2}^{diam(H_{k,n})}\frac{i-1}{2^{i}}=1-(diam(H_{k,n})+1)(\frac{%
1}{2})^{diam(H_{k,n})})  \label{summ1}
\end{equation}

    Observation 2. for $D_{0}$:\\
     If $n\not\equiv 1\, (\!\bmod\ {k}) $ then there
are $k$ vertices at distance $1,2,3,...,diam(H_{k,n})-1$ where $%
diam(H_{k,n})\!=\left\lceil \frac{n}{k}\right\rceil $ and there are $%
(n-1-(diam(H_{k,n})-1)k)$ vertices at distance $diam(H_{k,n}).$ If the vertex $0$ is
deleted, the distances between for each vertex $\frac{k}{2}(j-1)$ and the vertices $%
n-\frac{k}{2}i$ $\ $where $1\leq i\leq (diam(H_{k,n})-j)$ and $2\leq j\leq
diam-1$ increase by 1. Thus,

 $D_{0}=\sum\limits_{i=2}^{diam(H_{k,n})-1}\frac{%
i-1}{2^{i}}.$ Using equation (\ref{summ1}), we can express the summation as
\begin{equation}
\sum\limits_{i=2}^{diam(H_{k,n})-1}\frac{i-1}{2^{i}}=1-diam(H_{k,n})(\frac{1%
}{2})^{diam(H_{k,n})-1}.  \label{sum3}
\end{equation}%
Hence, utilizing equations (\ref{summ1}) and (\ref{sum3})%

\medskip
    If $n\equiv 1\; (\bmod\ k)$
\begin{equation*}
  R=\frac{n-2}{n}C(H_{k,n})-1+(\frac{1}{2})^{diam(H_{k,n})}\cdot(1+diam(H_{k,n}))
\end{equation*}

    If $n\not\equiv 1\; (\bmod\ k)$
\begin{equation*}
  R=\frac{n-2}{n}C(H_{k,n})-1+(\frac{1}{2})^{diam(H_{k,n})-1}\cdot diam(H_{k,n})
\end{equation*}
where $C(H_{k,n})$ is closeness value for $k$ is even. Therefore, the proof is completed.
\end{proof}
\begin{remark}
For $k=2$, the Harary Graph corresponds to the cycle graph. Therefore, the residual closeness value of $C_{n}$ is obtained in \cite{OdabasAytac}. Additionally, when $k$ is even and the diameter is $2$, in the case of $n \equiv 1 \pmod{k}$, Theorem \ref{diam2} yields the correct result. However, when $k$ is even and either $n \not\equiv 1 \pmod{k}$ or the diameter is $1$, Theorem \ref{diam2} does not apply. In these situations, the $D_{0}$ value vanishes. Therefore, the residual closeness value $R = \frac{n-2}{n}C(H_{k,n})$ is obtained.\end{remark}

\begin{theorem}
Let $H_{k,n}$ be a Harary Graph on $n$ vertices where $k>3$ is odd and $n$ is
even and $diam(H_{k,n})>2.$ Then, residual closeness value of $H_{k,n}$ can be expressed as follows:
\begin{itemize}
  \item If $n= (k-1)(2diam(H_{k,n})-1)+2$, then
\end{itemize}
\begin{equation*}
  R=\frac{n-2}{n}C(H_{k,n})-(1-\frac{(1+diam(H_{k,n}))}{2^{diam(H_{k,n})}})
\end{equation*}
\begin{itemize}
  \item If $n\neq(k-1)(2diam(H_{k,n})-1)+2$, then
\end{itemize}
\begin{equation*}
  R=\frac{n-2}{n}C(H_{k,n})-(1-\frac{diam(H_{k,n})}{2^{diam(H_{k,n})-1}}).
\end{equation*}
where $C(H_{k,n})$ closeness value of Harary Graph for $k$ is odd and $n$ is
even.
\end{theorem}
\begin{proof}
As in the previous theorem, the graph is vertex transitive. However, there are two observations in this case after removing a vertex from the graph. Without loss of generality, let the vertex $0$ be deleted from the graph. Thus, similar to the case when $k$ is an even number,%
\begin{eqnarray*}
C_{0} &=&C(H_{k,n})-2\sum\limits_{j\neq 0}\frac{1}{2^{d(0,j)}}-D_{0} \\
&=&\frac{n-2}{n}C(H_{k,n})-D_{0}
\end{eqnarray*}%
where $D_{0}$ denotes changes in closeness value of $C(H_{k,n})$ after
deleting $0$ and also, $C(H_{k,n})=n\sum\limits_{j\neq 0}\frac{1}{2^{d(0,j)}}$ for $k$ is odd, $n$ is even. It can be observed that when $n\neq(k-1)(2.diam(H_{k,n})-1)+2$, we can make city-tours with $diam(H_{k,n})$ moves and visiting first vertex $n/2$ traverse through city-tours with $diam(H_{k,n})$ moves from vertex $0$. The vertex $diam(H_{k,n})(\frac{k-1}{2})$ covered in $ diam(H_{k,n})^{th}$ city-tour at most from both direction. Nevertheless, when $n=(k-1)(2.diam(H_{k,n})-1)+2$, we can reach the vertex $diam(H_{k,n})(\frac{k-1}{2})$ with $diam(H_{k,n})+1$ moves visiting diametrically opposite vertex first. Therefore, we can examine the situation in two cases:

\medskip
Case 1. for $D_{0}$: Let $n\neq(k-1)\cdot(2diam(H_{k,n})-1)+2$. When
vertex $0$ is deleted, the diameter of the graph stays the same. Also, the distances between for each vertex that are $\frac{(k-1)(j-1)}{2}$ \ and the vertices $n-\frac{(k-1)}{2}i$ where $1\leq i\leq (diam(H_{k,n})-j)$ and $2\leq j\leq diam(H_{k,n})-1$ increase by 1. Therefore, changes in closeness value are
\begin{eqnarray*}
  D_{0}&=& \sum\limits_{i=2}^{diam(H_{k,n})-1}\frac{i-1}{2^{i}}\\
  &=&1-diam(H_{k,n})(\frac{1}{2})^{diam(H_{k,n})-1}
\end{eqnarray*}
from equation (\ref{sum3}).

\medskip
    Case 2. for $D_{0}$: Let $n=(k-1)\cdot(2diam(H_{k,n})-1)+2$. When
the vertex $0$ is deleted, the diameter of the graph increases by 1. Also, the distances between for each vertex $(k-1)(j-1)/2$ and the vertices $n-\frac{(k-1)}{2}i$  where $1\leq i\leq
(diam(H_{k,n})-j+1)$ and $2\leq j\leq diam(H_{k,n})$ increase by 1. Hence,
chances in closeness value is
\begin{eqnarray*}
D_{0}&=& \sum\limits_{i=2}^{diam(H_{k,n})}\frac{i-1}{2^{i}}\\
&=&1-(diam(H_{k,n})+1)(\frac{1}{2})^{diam(H_{k,n})})
\end{eqnarray*}
from equation (\ref{summ1}). Thus, the proof is done.\end{proof}

\begin{corollary}
Let $H_{k,n}$ be a Harary Graph on $n$ vertices where $k=3$ and $n$ is
even and \linebreak $diam(H_{k,n})>2.$ Then, residual closeness value of $H_{k,n}$ is
\begin{itemize}
  \item If $n\neq 4diam(H_{k,n})$
%\end{itemize}
\begin{eqnarray*}
  R &=& \frac{n-2}{n}C(H_{k,n})-\frac{3}{2}(1-diam(H_{k,n})(\frac{1}{2})^{(diam(H_{k,n})-1)}).
\end{eqnarray*}
%\begin{itemize}
  \item If $n=4diam(H_{k,n})$
%\end{itemize}
\begin{eqnarray*}
  R &=& \frac{n-2}{n}C(H_{k,n})-\frac{3}{2}+\frac{2diam(H_{k,n})+1}{2^{diam(H_{k,n})}}.
   \end{eqnarray*}
\end{itemize}
where $C(H_{k,n})$ closeness value of Harary Graph for $k=3$ and $n$ is
even.
\end{corollary}
\begin{proof} As proof of the previous theorem, when a vertex is removed from the graph, two situations are observed for $k=3$. Without loss of generality, let delete the vertex $0$ from the graph. Therefore, the total closeness value of a graph after removing vertex $0$ is
\begin{eqnarray*}
C_{0} &=&C(H_{k,n})-2\sum\limits_{j\neq 0}\frac{1}{2^{d(0,j)}}-D_{0} \\
&=&\frac{n-2}{n}C(H_{k,n})-D_{0}
\end{eqnarray*}
where $D_{0}$ denotes changes in closeness value of $C(H_{k,n})$ after
deleting $0$ and also $C(H_{k,n})=n\sum\limits_{j\neq 0}\frac{1}{2^{d(0,j)}}
$ for $k=3$, $n$ is even. It can be observed that when $n\neq(4diam(H_{k,n}))$, we can make city-tours with $diam(H_{k,n})$ moves and visiting first vertex $n/2$ traverse through city-tours with $diam(H_{k,n})$ moves from vertex $0$. The vertex $diam(H_{k,n})$ covered in $ diam(H_{k,n})^{th}$ city-tour from both direction. Nevertheless, when $n=(4diam(H_{k,n})$, we can reach the vertex $diam(H_{k,n})$ with $diam(H_{k,n})+1$ moves visiting diametrically opposite vertex first. Therefore, the residual closeness value can be obtained in the following two cases:\\

Case 1. for $D_{0}$: Let $n$ be distinct from $(4diam(H_{k,n}))$. When vertex $0$ is deleted, the distances between for each vertex $i$
and $(n-j)$ where $1\leq i\leq(diam(H_{k,n})-2)$ and $1\leq j\leq diam(H_{k,n})-i-1$ are increased by $2$. Hence, changes in closeness value are
\begin{eqnarray*}
D_{0}&=&2(\sum\limits_{i=2}^{diam(H_{k,n})-1}\frac{i-1}{2^{i}}-\frac{1}{4}\sum\limits_{i=2}^{diam(H_{k,n})-1}\frac{i-1}{2^{i}})\\
&=&\frac{3}{2}\sum\limits_{i=2}^{diam(H_{k,n})-1}\frac{i-1}{2^{i}}\\
&=&\frac{3}{2}(1-diam(H_{k,n})(\frac{1}{2})^{(diam(H_{k,n})-1)}).
\end{eqnarray*}
the result of $D_{0}$ is obtained from equations (\ref{sum3}).
\\

Case 2. for $D_{0}$: Let $n=(4diam(H_{k,n}))$ . When
vertex 0 is deleted, the distances between for each vertices $i$
and $(n-j)$ where $1\leq i\leq(diam(H_{k,n})-1)$ and $1\leq j\leq diam(H_{k,n})-i-1$ are increased by 2. Additionally, distance between vertices $i$ and $n-diam(H_{k,n})+i$, where $1\leq i\leq(diam(H_{k,n})-1)$, is increased by 1. Hence, changes in closeness values are
\begin{eqnarray*}
D_{0}&=&2\cdot[\sum\limits_{i=2}^{diam(H_{k,n})-1}\frac{i-1}{2^{i}}-\frac{1}{4}\sum\limits_{i=2}^{diam(H_{k,n})-1}\frac{i-1}{2^{i}}+(\frac{diam(H_{k,n})-1}{2^{diam(H_{k,n})}}-\frac{diam(H_{k,n})-1}{2^{diam(H_{k,n})+1}}) ]\\
&=& \frac{3}{2}\cdot\sum\limits_{i=2}^{diam(H_{k,n})-1}\frac{i-1}{2^{i}}+\frac{diam(H_{k,n})-1}{2^{diam(H_{k,n})}}\\
&=&\frac{3}{2}\cdot(1-diam(H_{k,n})(\frac{1}{2})^{(diam(H_{k,n})-1)})+\frac{diam(H_{k,n})-1}{2^{diam(H_{k,n})}}\\
&=&\frac{3}{2}-\frac{2diam(H_{k,n})+1}{2^{diam(H_{k,n})}}.
\end{eqnarray*}
the result of $D_{0}$ is obtained from equations (\ref{sum3}).
Thus, the proof is done.
\end{proof}
\begin{theorem} \label{RCkoddnodd}
Let $H_{k,n}$ be a Harary Graph on $n$ vertices where $k>3$ and $n$ are odd
and \linebreak $diam(H_{k,n})>2.$ Then, residual closeness value of $H_{k,n}$
\end{theorem}
\begin{equation*}
R=C(H_{k,n})-3k-2+(\frac{4k-3-t+3diam(H_{k,n})}{2^{diam(H_{k,n})-1}}).
\end{equation*}%
where $(n-k-1)\equiv t (\bmod\ (2(k-1)))$ and $C(H_{k,n})$ is closeness
value of Harary Graph for $k$ and $n$ is odd.

\begin{proof}
In this case, graph is not vertex transitive. We want to get a minimum
closeness value of $\sum\limits_{\substack{ i\neq v, \\ \forall i\in
V(H_{k,n})}}C_{v}(i)$ when vertex $v$ removed from the graph. Let us consider
the closeness values after removing the vertex according to the types.\\

Observation 1. If $v=\frac{n-1}{2}$ removed from the graph, then
\begin{equation*}
C_{v}(i)=C(H_{k,n})-2C(v)-D_{v}
\end{equation*}%
where $D_{v}$ denotes changes in the closeness value when $v=\frac{n-1}{2}$
removed from the graph. When vertex $v$ is deleted, the distances of
vertices $\frac{(n-1)}{2}-(\frac{k-1}{2})i$ and $\frac{(n-1)}{2}+(\frac{k-1}{%
2})i$ to vertices $(\frac{k-1}{2})j$, $n-(\frac{k-1}{2})j-1$ increased by $1$%
, respectively where $1\leq i\leq (diam(H_{k,n})-2)$ and $0\leq j\leq
(diam(H_{k,n})-(i+2))$. In addition, the distances between vertices $\frac{%
(n-1)}{2}-(\frac{k-1}{2})i$ and $\frac{(n-1)}{2}+(\frac{k-1}{2})j$ where $%
1\leq i\leq (diam(H_{k,n})-2)$ and $1\leq j\leq (diam(H_{k,n})-(i+1))$ are also increased by $1$. Therefore, changing is expressed as $\sum%
\limits_{i=2}^{diam(H_{k,n})-1}\frac{i-1}{2^{i}}.$  If we consider these
three kinds of situations then the total closeness value $6\sum\limits_{i=2}^{diam(H_{k,n})-1}\frac{i-1}{2^{i}}$
modifies to $\frac{1}{2}\cdot6\sum\limits_{i=2}^{diam(H_{k,n})-1}\frac{i-1}{2^{i}%
}$. Therefore, the changes in closeness value can be expressed as $ \\
D_{v}=3\sum\limits_{i=2}^{diam(H_{k,n})-1}\frac{i-1}{2^{i}}.$ Hence we have,\vspace*{-3mm}
\begin{equation*}
C_{v}(i)=C(H_{k,n})-2C(\frac{n-1}{2})-3\sum\limits_{i=2}^{diam(H_{k,n})-1}%
\frac{i-1}{2^{i}}
\end{equation*}%
Using equation (\ref{sum3}), we get
\begin{equation}
=C(H_{k,n})-2C(\frac{n-1}{2})-3(1-diam(H_{k,n})(\frac{1}{2}%
)^{diam(H_{k,n})-1}).  \label{resn-1/2}
\end{equation}

Observation 2. If one of the $v=\frac{n-1}{2}\pm i,$ $1\leq i\leq \frac{k-1}{%
2}$ vertices removed from the graph then two situations will be appeared.
\begin{equation*}
C_{v}(i)=C(H_{k,n})-2C(v)-D_{v}
\end{equation*}%
where $D_{v}$ denotes changes in the closeness value when the vertex $v$ is
removed from the graph. This observation can be separated into two cases:

\medskip
Case 1.  When one of the $v=\frac{n-1}{2}\pm i,~1\leq i\leq \frac{k-3}{2}$
removed from the graph:\\ Assume that for $i=1$ the corresponding vertex $v=%
\frac{n-1}{2}+1=\frac{n+1}{2}$ is deleted. After the vertex $v=\frac{n+1}{2}$
is removed from the graph;
the distances between vertices $v+(\frac{k-1}{2})m$ and $v-(\frac{k-1}{%
2})j$ which are connected with $2,3,...,(diam(H_{k,n})-1)$ city-tour, where $%
1\leq m\leq diam(H_{k,n})-2$ and $1\leq j\leq diam(H_{k,n})-(m+1),$ are
increased by $1$. Therefore, the value $\sum\limits_{i=2}^{diam(H_{k,n})-1}%
\frac{i-1}{2^{i}}$ is half-modified.

\medskip
    Analogously, the distance between vertices $v-(\frac{k-1}{2})m$ and $%
(v+\frac{n-1}{2})_{\mod(n-1)}+(\frac{k-1}{2})j$ where $1\leq m\leq
diam(H_{k,n})-2$ and $0\leq j\leq diam(H_{k,n})-(m+1)$ are increased by $1.$
Here, the vertex $(v+\frac{n-1}{2})_{\mod(n-1)}$ is diametrically
opposite vertex of $v$. Therefore, the value $\sum%
\limits_{i=2}^{diam(H_{k,n})-1}\frac{i-1}{2^{i}}$ is half-modified. This
implies that
\begin{equation}
C_{v}(i)=C(H_{k,n})-2C(\frac{n+1}{2})-2\sum%
\limits_{i=2}^{diam(H_{k,n})-1}\frac{i-1}{2^{i}}.  \label{resn-1/2+i}
\end{equation}%
The same modification should be considered for the reverse order of vertex
pairs that have increased the distance between them by $1$. Therefore, the
coefficient of variation was multiplied by $4$. Equation (\ref{sum3}) can be
substituted as $\sum\limits_{i=2}^{diam(H_{k,n})-1}\frac{i-1}{2^{i}}%
=(1-diam(H_{k,n})(\frac{1}{2})^{diam(H_{k,n})-1})$ in equation (\ref%
{resn-1/2+i}).%
\begin{eqnarray}\label{14}
C_{v}(i)&=&C(H_{k,n})-2C(\frac{n+1}{2})-2(1-(diam(H_{k,n})\cdot(\frac{1}{2})^{diam(H_{k,n})-2}))\nonumber.
\end{eqnarray}

\medskip
Case 2. When one of the $v=\frac{n-1}{2}\pm \frac{k-1}{2}$ removed from the
graph:\\
 Assume that the vertex $v=\frac{n-1}{2}-\frac{k-1}{2}$ is deleted.
After the vertex $v$ is removed from the graph, the distances between vertices $v+(\frac{k-1}{2})m$ and $v-(\frac{k-1}{2})j$
which are connected with $2,3,...,(diam(H_{k,n})-1)$ city-tour, where $%
1\leq m\leq diam(H_{k,n})-2$ and $1\leq j\leq diam(H_{k,n})-(m+1),$
are increased by~$1.$ Therefore, the value $\sum\limits_{i=2}^{diam(H_{k,n})-1}%
\frac{i-1}{2^{i}}$ is half-modified.

\medskip
    Analogously, the distance between vertices $v-(\frac{k-1}{2})m$ and $%
(v+\frac{n-1}{2})_{\mod(n-1)}+(\frac{k-1}{2})j$, where $1\leq m\leq
diam(H_{k,n})-3$ and $0\leq j\leq diam(H_{k,n})-(m+2)$, is increased by $1.$
Therefore, the value $\sum\limits_{i=2}^{diam(H_{k,n})-1}\frac{i-1}{2^{i}}$
is half-modified. This implies that; the distance between vertices $v+\frac{%
k-1}{2}$ and $v-\frac{k-1}{2}$ is increased by $1.$ Since the distance between the vertex $v-(\frac{k-1}{2})$ and vertex $0$ is three before removing the vertex $v.$ This implies that
\begin{eqnarray}\label{res++i}
C_{v}(i)=C(H_{k,n})-2C(\frac{n-1}{2}-\frac{k-1}{2})-\sum\limits_{i=2}^{diam(H_{k,n})-1}\frac{i-1}{2^{i}}-\sum\limits_{i=3}^{diam(H_{k,n})-1}\frac{i-1}{2^{i}}
\end{eqnarray}
Using geometric sum formula (\ref{geo}) and its derivative (\ref{derigeo}), the equation (\ref{res++i}) evaluated as in equation (\ref{sum3})%
\begin{eqnarray}\label{res+-i}
C_{v}(i)&=&C(H_{k,n})-2C(\frac{n-1}{2}-\frac{k-1}{2})-(1-(diam(H_{k,n})(\frac{1}{2})^{diam(H_{k,n})-2}))+\frac{1}{4}.
\end{eqnarray}

Observation 3. If $v\in RM$ removed from graph, then
\begin{equation*}
C_{v}(i)=C(H_{k,n})-2C(v)-D_{v}
\end{equation*}%
where $D_{v}$ denotes changes in the closeness value when $v$ is removed
from the graph. Let the vertex $0$ is removed from the graph. The distance
between vertices $\frac{k-1}{2}m$ and $(\frac{k-1}{2})_{\mod n}+(\frac{%
k-1}{2})j$ where $1\leq m\leq diam(H_{k,n})-2$ and $1\leq j\leq
diam(H_{k,n})-(m+1),$ is increased by $1.$ This implies that
\begin{equation*}
C_{v}(i)=C(H_{k,n})-2C(0)-\sum\limits_{i=2}^{diam(H_{k,n})-1}\frac{i-1}{2^{i}}
\end{equation*}%
Using summation
%TCIMACRO{\TeXButton{TeX field}{\ref{sum3}}}%
%BeginExpansion
(\ref{sum3}), the following equation can be expressed as
\begin{equation}\label{resn-1/2+iii}
=C(H_{k,n})-2C(0)-(1-(diam(H_{k,n})(\frac{1}{2})^{diam(H_{k,n})-1}))
\end{equation}%
for any vertex $v$ in the set of $RM.$  In order to get minimum residual
value of $H_{k,n}$ when $k$ and $n$ are odd, it is needed to compare
closeness value of vertices labeled by $\frac{n-1}{2},$ $\frac{n-1}{2}\pm i
$ , for $1\leq i\leq \frac{k-1}{2}$ and remaining vertices in the set $RM$ first. We have the closeness value of a vertices from equations
(\ref{eq1}),
%TCIMACRO{\TeXButton{TeX field}{\ref{eq2}}}%
%BeginExpansion
(\ref{eq2})%
%EndExpansion
, and
%TCIMACRO{\TeXButton{TeX field}{\ref{eq3}} }%
%BeginExpansion
(\ref{eq3})
%EndExpansion
where $(n-k-1)\equiv t \mod 2(k-1).$  It can be said from proof of Theorem \ref{nkodd} that $C(\frac{n-1}{2})$ is the maximum closeness value
among other closeness values of vertices. In order to obtain minimum value
after a vertex is removed from the graph, let us compare the values of equations
(\ref{resn-1/2}),
(\ref{14}),
(\ref{res+-i}),
and (\ref{resn-1/2+iii}).
%TCIMACRO{\TeXButton{TeX field}{\ref{resn-1/2+ii}}}%
%BeginExpansion
%
%EndExpansion
If we compare the change values after removing the vertex, denoted by $D_{v},$
$C_{v}(\frac{n-1}{2})$ corresponds to residual closeness value due to the
maximum closeness value comes from $C(\frac{n-1}{2}).$  Therefore, we get
\begin{equation*}
R=C(H_{k,n})-2C(\frac{n-1}{2})-3(1-diam(H_{k,n})(\frac{1}{2}%
)^{diam(H_{k,n})-1})
\end{equation*}%
where
\begin{eqnarray*}
C(\frac{n-1}{2})&=&\frac{k+1}{2}+(k-1)(1-(\frac{1}{2}%
)^{diam(H_{k,n})-2})+\frac{t-1}{2^{diam(H_{k,n})}}.
\end{eqnarray*}
Hence, we obtain the residual value simply as%
\begin{equation*}
R=C(H_{k,n})-3k-2+(\frac{4k-3-t+3diam(H_{k,n})}{2^{diam(H_{k,n})-1}}).
\end{equation*}%
Then, the proof is completed.
\end{proof}
If $k = 3$, this case will be considered separately due to the distinct changes in distances when a vertex is removed from the graph. Relying on the proof of Theorem \ref{RCkoddnodd}, the smallest residual value for the special case of $k = 3$ will also be obtained by removing the vertex $v = \frac{n-1}{2}$. The result for the case $k = 3$ is as follows:
\begin{corollary}\label{corollary3.6}
	Let $H_{k,n}$ be a Harary Graph on $n$ vertices where $k=3$ and $n$ are odd, $diam(H_{k,n})>3$, and $(n-4)\equiv t (\bmod\ 4)$ Then, residual closeness value of $H_{k,n}$:\\	
		If $t= 3$, then
	\begin{equation*}
		R=C(H_{k,n})-11,75+\frac{8diam(H_{k,n})+11}{2^{diam(H_{k,n})}}.
	\end{equation*}
 If $t= 1$, then
	\begin{equation*}
		R=C(H_{k,n})-11,75+\frac{9 diam(H_{k,n})+13}{2^{diam(H_{k,n})}}.
	\end{equation*}%
		where $C(H_{k,n})$ is closeness value of Harary Graph for $k$ and $n$ is odd.
	\end{corollary}
\begin{proof}
	In this case, graph is not vertex transitive. We want to get minimum
	closeness value of $\sum\limits_{\substack{ i\neq v, \\ \forall i\in
			V(H_{k,n})}}C_{v}(i)$ when vertex $v = \frac{n-1}{2}$ removed from the graph as in the proof of Theorem \ref{RCkoddnodd}. Additionally, we will examine the case for $k = 3$ in two different scenarios based on whether the value of $t$  is equivalent to 1 or 3:
		
\medskip
	Case 1. Let $(n-4)\equiv 3 (\bmod\ 4)$
	\begin{equation*}
		C_{v}(i)=C(H_{k,n})-2C(v)-D_{v}
	\end{equation*}%
	where $D_{v}$ denotes changes in the closeness value when $v=\frac{n-1}{2}$ removed from the graph. When vertex $v$ is deleted, the distances of vertices $\frac{(n-1)}{2}-i$ and $\frac{(n-1)}{2}+i$ to the vertices $j$, $n-j-1$ is increased by $1$  where $1\leq i\leq (diam(H_{k,n})-2)$ and $0\leq j\leq	(diam(H_{k,n})-(i+2))$. Therefore, the total closeness value before vertex removal can be expressed as $\sum%
	\limits_{i=2}^{diam(H_{k,n})-1}\frac{i-1}{2^{i}}$ for each indicated vertex pairs. If we consider these four kinds of situations then total change between vertices can be calculated as follows:
	
\medskip
Total closeness value $(4\sum\limits_{i=2}^{diam(H_{k,n})-1}\frac{i-1}{2^{i}})$
	modifies to  $(\frac{1}{2}\cdot4\sum\limits_{i=2}^{diam(H_{k,n})-1}\frac{i-1}{2^{i}})$.  Therefore, the \linebreak
total changes in closeness value can be expresses as 	$(2\sum\limits_{i=2}^{diam(H_{k,n})-1}\frac{i-1}{2^{i}}).$\\

In addition, the distances between vertices $\frac{	(n-1)}{2}-i$ and $\frac{(n-1)}{2}+j$ where $1\leq i\leq (diam(H_{k,n})-2)$ and $1\leq j\leq (diam(H_{k,n})-(i+1))$ also increases. The changing can be evaluated as follows:\\

Before removal, the total closeness value between indicated vertex pairs can be expressed as $(2\sum\limits_{i=2}^{diam(H_{k,n})-1}\frac{i-1}{2^{i}})$ and after removal, the total closeness value between indicated vertex pairs can expressed as $(2\cdot(\sum\limits_{i=2}^{diam(H_{k,n})-2}\frac{i-1}{2^{i+3}}+\frac{diam(H_{k,n})-2}{2^{diam(H_{k,n})+1}}))$.
	 Thus,  the  $D_{v}=(\frac{15}{4}\sum\limits_{i=2}^{diam(H_{k,n})-1}\frac{i-1}{2^{i}}) -$  $(\frac{diam(H_{k,n})-2}{2^{diam(H_{k,n})+1}})$ value is obtained.
	
\medskip
	 Hence, using equation (\ref{sum3})
	\begin{equation*}
		C_{v}(i)=C(H_{k,n})-2C(\frac{n-1}{2})-(\frac{15}{4}-\frac{16 diam(H_{k,n})-2}{2^{(diam(H_{k,n})+1)}})
	\end{equation*}%
	Substituting the equation (\ref{eq1}) of the form $k=3$ and $t=3$, we obtain the residual closeness value as:
	\begin{eqnarray}
%		&=&C(H_{k,n})-(11,75-\frac{17-2t+8diam(H_{k,n})}{2^{diam(H_{k,n})}})\\
		=C(H_{k,n})-11,75+\frac{8diam(H_{k,n})+11}{2^{diam(H_{k,n})}}.
	\end{eqnarray}
	
\smallskip
Case 2. Let $(n-4)\equiv 1 (\bmod\ 4)$
	\begin{equation*}
		C_{v}(i)=C(H_{k,n})-2C(v)-D_{v}
	\end{equation*}%
	where $D_{v}$ denotes changes in the closeness value when $v=\frac{n-1}{2}$ removed from the graph. When vertex $v$ is deleted, as in Case 1, the distances of vertices $\frac{(n-1)}{2}-i$ and $\frac{(n-1)}{2}+i$ to the vertices $j$, $n-j-1$ are increased by $1$  where $1\leq i\leq (diam(H_{k,n})-2)$ and $0\leq j\leq	(diam(H_{k,n})-(i+2))$. Therefore, the total closeness value before vertex removal can be expressed as $\sum%
	\limits_{i=2}^{diam(H_{k,n})-1}\frac{i-1}{2^{i}}$ for each indicated vertex pairs. If we consider these four kinds of situations then total change between vertices can be calculated as follows:\\
	
	Total closeness value between indicated pairs $(4\sum\limits_{i=2}^{diam(H_{k,n})-1}\frac{i-1}{2^{i}})$
	modify to \\  $(\frac{1}{2}\cdot4\sum\limits_{i=2}^{diam(H_{k,n})-1}\frac{i-1}{2^{i}%
	})$. Therefore, the total changes in closeness value can be expressed as \linebreak 	$(2\sum\limits_{i=2}^{diam(H_{k,n})-1}\frac{i-1}{2^{i}}).$\\
	
	In addition, the distances between vertices $\frac{	(n-1)}{2}-i$ and $\frac{(n-1)}{2}+j$ where $1\leq i\leq (diam(H_{k,n})-2)$ and $1\leq j\leq (diam(H_{k,n})-(i+1))$ also increases. The changing can be evaluated as follows:\\
	
	Before removal, the total closeness value between indicated vertex pairs can be expressed as $(2\sum\limits_{i=2}^{diam(H_{k,n})-1}\frac{i-1}{2^{i}})$, and after removal, the total closeness value between indicated vertex pairs can expressed as $(2\cdot(\sum\limits_{i=2}^{diam(H_{k,n})-2}\frac{i-1}{2^{i+3}}+\frac{diam(H_{k,n})-2}{2^{diam(H_{k,n})}}))$.
	Thus,   $D_{v}=(\frac{15}{4}\sum\limits_{i=2}^{diam(H_{k,n})-1}\frac{i-1}{2^{i}}) \; -$ \linebreak $(\frac{3diam(H_{k,n}-6)}{2^{diam(H_{k,n})+1}})$ value is obtained.
	
\medskip
	Hence, using equation (\ref{sum3}), we get
	\begin{equation*}
		C_{v}(i)=C(H_{k,n})-2C(\frac{n-1}{2})-(\frac{15}{4}-\frac{9 diam(H_{k,n})-3}{2^{diam(H_{k,n})}})
	\end{equation*}
	Substituting the equation (\ref{eq1}) of the form $k=3$ and $t=1$, we obtain the residual closeness value as:
	\begin{equation}
		R=C(H_{k,n})-11,75+\frac{9 diam(H_{k,n})+13}{2^{diam(H_{k,n})}}.
	\end{equation}

\vspace*{-3mm}
\end{proof}

\begin{remark}
	Due to the constraints established in the proof of Corollary \ref{corollary3.6}, the diameter value was considered to be greater than $3$. Thus, for odd values of $n$ and $k=3$, residual closeness values for diameter values $2$ and $3$ were not formulated. To ensure completeness in the calculations, we can provide the residual closeness values for the remaining four specific small Harary Graph structures as follows: $R(H_{3,5})=5, R(H_{3,7})=11, R(H_{3,9})=17.5, R(H_{3,11})=24.875$.
	\end{remark}

\end{document}